\documentclass[runningheads]{llncs}
\usepackage[T1]{fontenc}
\usepackage{graphicx}
\usepackage{hyperref}
\usepackage{color}

\usepackage{verbatim}
\usepackage{amsmath}
\usepackage[capitalize]{cleveref}
\crefname{algocf}{Algorithm}{algorithms}
\crefname{observation}{Observation}{observations}

\newcommand{\defproblem}[3]{
  \vspace{2mm}
\noindent\fbox{
  \begin{minipage}{0.94\textwidth}
  \textsc{#1}\\
  {\bf{Input:}} #2  \\
  {\bf{Output:}} #3
  \end{minipage}
  }
  \vspace{2mm}
}

\spnewtheorem{observation}{Observation}{\itshape}{\rmfamily}

\usepackage{todonotes}

\usepackage{tikz}
\usetikzlibrary{fit,patterns,decorations.pathmorphing,decorations.pathreplacing,calc}

\usepackage[ruled,noend,linesnumbered]{algorithm2e}
\SetKw{Continue}{continue}
\SetKw{Break}{break}

\newcommand{\dd}{\mathinner{.\,.}}

\newcommand{\fact}{\mathcal{F}}

\newcommand{\dol}{\mathtt{\$}}
\newcommand{\factT}{\mathcal{F}_T}

\newcommand{\false}{\mathbf{false}}

\newcommand{\BWT}{\mathtt{BWT}}
\newcommand{\SA}{\mathtt{SA}}
\newcommand{\ISA}{\mathtt{ISA}}

\newcommand{\LCE}{\mathtt{LCE}}
\newcommand{\LCP}{\mathtt{LCP}}
\newcommand{\lcp}{\mathit{lcp}}

\newcommand{\rev}[1]{\overline{#1}}
\newcommand{\TR}{\overline{T}}
\newcommand{\BWTR}{\overline{\mathtt{BWT}}}
\newcommand{\SAR}{\overline{\mathtt{SA}}}
\newcommand{\LCER}{\overline{\mathtt{LCE}}}
\newcommand{\LCPR}{\overline{\mathtt{LCP}}}
\newcommand{\LFR}{\overline{\mathtt{LF}}}
\newcommand{\ISAR}{\overline{\mathtt{ISA}}}
\newcommand{\rR}{\overline{r}}

\newcommand{\sfxtset}{\mathcal{X}}
\newcommand{\SRE}{\mathtt{SRE}}

\newcommand{\colt}{<_{\mathtt{co}}}

\begin{document}
\title{Computing Smallest Suffixient Arrays \\ in Sublinear Time}
%
%
\author{Hiroto Fujimaru\inst{1} \and Gonzalo Navarro\inst{2,3} \and Francisco Olivares\inst{2,3} \and Jakub Radoszewski\inst{4} \and Giuseppe Romana\inst{5} \and Cristian Urbina\inst{4}}
\authorrunning{H. Fujimaru et al.}
%
\institute{Department of Information Science and Technology, Kyushu University, Japan \email{fujimaru.hiroto.134@s.kyushu-u.ac.jp}\and Department of Computer Science, University of Chile, Chile \email{\{gnavarro,folivares\}@uchile.cl}\and Center for Biotechnology and Bioengineering (CeBiB),  Chile  \and
Faculty of Mathematics, Informatics and Mechanics, University of Warsaw, Poland \email{\{j.radoszewski,c.urbina-gallegos\}@uw.edu.pl}\and Faculty of Mathematics, University of Palermo,  Italy\\ \email{giuseppe.romana01@unipa.it}}
\maketitle              
\begin{abstract}
A suffixient array is a novel data structure that, when combined with an index providing direct access on a text $T$, allows us to answer a variety of pattern matching queries. 
In this work, we show how to compute a smallest suffixient array for $T[1\dd n]$ in $O(\frac{n\log \sigma}{\sqrt{\log n}}+\min(r,\rR)\log^\epsilon n)$ time for any $\epsilon > 0$, where $\sigma$ is the alphabet size of $T$ and $r$ and $\rR$ are the numbers of equal-letter runs of the Burrows-Wheeler transforms of $T$ and its reverse $\TR$, respectively. This time complexity becomes sublinear when $\sigma$ is small enough and $\min(r,\rR)=o(\frac{n}{\log^\epsilon n})$, yielding an asymptotic improvement over state-of-the-art algorithms. We also present a series of connected algorithmic results.
\end{abstract}

\section{Introduction}

One of the most relevant challenges in modern data compression is to represent in compressed space the huge and highly repetitive text collections that naturally arise in fields like Bioinformatics \cite{Genome1k,MediniDTMR05,Wang22}, so as to query them directly in compressed form. To do so, several measures of repetitiveness have been proposed to evaluate the compression effectiveness of those representations \cite{Nav22:survey-a,Nav22:survey-b}. Among others, {\em suffixient arrays} \cite{DepuydtGLMP:ArXiv23,cenzato2025suffixientarraysnewefficient} have received much attention in recent years. A suffixient array is a data structure that, once built on a text $T[1\dd n]$, and provided with a representation of $T$ that provides efficient random access to it, can efficiently find one occurrence of $P$ in $T$, and more generally, of each maximal substring of $P$ that occurs in $T$ (called ``maximal exact matches'', or ``MEMs''). This functionality is inferior to that of suffix arrays, but in exchange suffixient arrays can be much smaller on repetitive text collections, and they were shown to be faster in practice than than other well-known repetitiveness-aware compressed suffix arrays, such as the $r$-index \cite{GNP20}.

More concretely, the size of a suffixient array is linear in the size of a {\em suffixient set} of $T$, a subset of the positions of $T$ such that any right-extension (i.e., one-character extension of a right-maximal substring) appears aligned at its rightmost position with some suffixient set position \cite{DepuydtGLMP:ArXiv23}. In a way, a suffixient set captures all the different substrings appearing in the text. It is natural then to aim for suffixient sets of minimum cardinality \cite{CenzatoOP:spire23}, which has been called $\chi$ and studied as a repetitiveness measure in its own right \cite{cenzato2025suffixientarraysnewefficient,NRU_SPIRE_2025,FNRU_arxiv,chi_reachable}.

Besides the study of $\chi$, a relevant challenge is how to efficiently find smallest suffixient sets and compute their  suffixient arrays (which are the lexicographic order of the reversed prefixes that end at the positions of the suffixient sets). Cenzato et al.\ showed how to find smallest suffixient sets in linear time and compressed space \cite{CenzatoOP:spire23}, and how to test if a given set is suffixient \cite{CenzatoOP:arXiv25}. Other constructions are linear-time online \cite{FNRU_arxiv}, or near-real-time \cite{koppl2026smallestsuffixientsetmaintenance}. The smallest suffixient arrays can also be constructed in linear time \cite{DBLP:conf/cpm/BonizzoniGR26,cenzato2025suffixientarraysnewefficient}.

In this paper we present the first {\em sublinear-time} algorithm to build smallest suffixient arrays under the RAM model of computation (where computer words of $\Omega(\log n)$ bits are manipulated in constant time). We assume that the input text, which is over an alphabet $[0\dd\sigma)$, comes packed into $O((n\log \sigma)/\log n)$ words, and build a smallest suffixient array in time $O(n\log\sigma/\sqrt{\log n} + r\log^\epsilon n)$ for any $\epsilon>0$, where $r$ is the number of runs in the Burrows-Wheeler Transform (BWT) \cite{BW94} of $T$. Measure $r \le n$ is also sensitive to repetitiveness; in particular it holds that $\chi \le 2r$ \cite{NRU_SPIRE_2025}. More precisely, our contributions are as follows:

\begin{enumerate}
\item We give a new linear-time algorithm for computing a smallest suffixient array starting from the BWT, the suffix array, and the longest common prefix array of $T$. This offers an alternative to the algorithms of Cenzato et al.~\cite{cenzato2025suffixientarraysnewefficient}, which require those structures built on the reversed text instead. Our algorithm uses less working space than Fujimaru et al.'s \cite{FNRU_arxiv} online algorithm for constructing suffixient sets. Further, it is very intuitive and easy to implement. As a first application, we use our algorithm to build suffixient arrays in compressed working space $O((r+\rR)\log n)$, and time $O(n(\log (r+\rR)+\log\log n))$, if the BWT of the reversed text has $\rR$ runs.

\item We optimize our algorithm to work in sublinear time under the packed setting, using recent data structures 
\cite{sublinear_BWT_construction_fullversion,sublinear_suffix_array_construction} for sublinear time computation of the BWT, longest common extensions, and compressed suffix trees. We then compute a smallest suffixient array for $T$ in $O(n\log\sigma/\sqrt{\log n} + r\log^\epsilon n)$ time, which 
is sublinear when $\log\sigma = o(\sqrt{\log n})$ and $r = o(n/\log^\epsilon n)$ for some $\epsilon>0$. The working space is $O(n\log\sigma/\log n + r)$.

\item We show that, when $\log\sigma=o(\sqrt{\log n})$, this result yields an $o(n)$-time algorithm to build a suffixient array whose size is an $O(\log^{1+\epsilon} n)$-approximation to $\chi$. Further, size $O(\chi)$ is not reached only when $\chi = \Omega(n/\log^{\epsilon} n)$, that is, when the text is not so repetitive and thus it is not very interesting to have a suffixient array instead of a plain suffix array. By weaking the approximation factor to $O(\log^{3/2+\epsilon}n)$, we reach the base construction time $O(n\log\sigma/\sqrt{\log n})$.

\item We adapt Algorithm 6 of Cenzato et al.~\cite{cenzato2025suffixientarraysnewefficient} to the packed setting, achieving the same sublinear time and working space on the reverse text.

\item Combining both results, we achieve $O(n\log\sigma/\sqrt{\log n} + \min(r,\overline{r})\log^\epsilon n)$-time computation of a smallest suffixient array for $T$, within working \sloppy space $O(n\log\sigma/\log n + \min(r,\overline{r}))$. Since $r = \Omega(\overline{r} \log n)$ on some string families \cite{GiulianiILPST21}, this symmetry yields strictly better complexity.

\item As a byproduct, we improve the $O(n)$-time construction of Shibata and Bannai~\cite{chi_reachable} of an $O(\chi(T))$-sized string representation to $O(n\log\sigma/\sqrt{\log n} + r\log^\epsilon n)$ time and $O(n\log\sigma/\log n + r)$ space.
\end{enumerate}

\section{Terminology}

We consider an alphabet $\Sigma=[0, \sigma)$ of size $\sigma$. A \emph{string} $x[1\dd n]$ of length $n$ is a finite sequence $x[1]x[2]\cdots x[n]$ of $n$ elements in $\Sigma$. The \emph{empty string} $\varepsilon$ is the unique string of length $0$. We let $\Sigma^*$ be the set of all strings over $\Sigma$, $\Sigma^+=\Sigma^*\setminus \{\varepsilon\}$, and $\Sigma^k$ the set of strings of length $k$ for $k \ge 0$. We let $x[i\dd j]=x[i]x[i+1]\cdots x[j]$ if $1\le i\le j \le n$; if $j < i$, we let $x[i\dd j] = \varepsilon$, and use curved brackets to exclude the borders, e.g. $x(i\dd j)=x[i+1]\cdots x[j-1]$. The concatenation of two strings $x[1\dd n]$ and $y[1\dd m]$ is the string $x \cdot y = xy = x[1]\cdots x[n]y[1]\cdots y[m]$. A \emph{substring} (or \emph{factor}) of $x$ is any string $y$ such that $x=vyz$ for some $v, z\in \Sigma^*$. A \emph{prefix} (resp. \emph{suffix}) of $x$ is any substring $y$ such that $x=yz$ (resp. $x =zy$) for some string $z \in \Sigma^*$. We denote $\fact_x(k)$ the set of substrings of length $k$ of $x$ for any $k\ge 0$, and $\fact_x$ the set of substrings of $x$ of any possible length. The \emph{reverse} of a string $x[1\dd n]$, denoted as $\rev{x}$, is the string $x[n]x[n-1]\cdots x[1]$. The \emph{lexicographic order} between two strings is defined as follows: $x < y$ if $y=xz$ with $z\in\Sigma^+$, or $x=zax'$ and $y =zby'$ for some $z, x',y'\in \Sigma^*$ and $a,b\in \Sigma$ such that $a < b$. The \emph{co-lexicographic order} is defined as $x \colt y \iff \rev{x} < \rev{y}$.

We denote a string $T[1\dd n]$ and call it a \emph{text} to emphasize it is the target string to be compressed or indexed. Further, we assume $T[n]=\dol$, where $\dol$ is a special symbol that appears (only) at the end of $T$ and is smaller than any other symbol in $\Sigma$. 
We use lowercase  letters like $x,y,z$ to denote arbitrary strings.

The \emph{suffix array (SA)} of $T[1\dd n]$ is an array $\SA[1\dd n]$ such that $\SA[i]=j$ if $T[j\dd n]$ is the $i$-th suffix of $T$ in lexicographic order. By $\ISA[1 \dd n]$ we denote the inverse permutation of the suffix array, that is, an array such that $\ISA[i]=j \iff \SA[j]=i$ for $i,j \in [1,n]$.

The \emph{longest common prefix} of two strings $x$ and $y$, denoted $\lcp(x,y)$, is the longest string that is a prefix of $x$ and $y$. A \emph{longest common extension query} $\LCE(i, j)$ returns the length of the longest common prefix between the suffixes $T[i\dd n]$ and $T[j \dd n]$. The \emph{longest common prefix array} is defined as $\LCP[i]=\LCE(\SA[i],\SA[i-1])$ for $i > 1$ and $\LCP[1]=0$.

The \emph{Burrows--Wheeler transform (BWT)} \cite{BW94} of a text $T$ is a permutation obtained by sorting all the suffixes of $T$ in lexicographic order and concatenating their preceding symbols. Formally, $\BWT(T) = T[\SA[1]-1] \cdots T[\SA[n]-1]$ where we let $T[0]=T[n]$. For any string $x$, its \emph{run-length encoding} is the unique sequence of pairs $(a_1,p_1),\ldots,(a_t,p_t)$ such that $x=a_1^{p_1}\cdots a_t^{p_t}$ and $a_{i-1}\neq a_{i}$ for all $i \in [2,t]$. In particular, we use $r$ to denote the size of the run-length encoding of $\BWT(T)$, which is considered a measure of the repetitiveness of the text  \cite{Nav22:survey-a}. A position $i$ in $\BWT(T)$ is said to be a \emph{BWT run-break} if $\BWT[i-1]\neq \BWT[i]$. 

We use $\BWTR,\SAR,\LCPR,\ISAR,\LCER$ and $\rR$, to denote the corresponding data structures and values for $\TR$ instead of $T$.

\subsection{Suffixient sets and supermaximal right extensions}

The set of \emph{right extensions}  of $T$ is $E_T=\{xa\in\factT\, |\,\exists b \neq a \text{ such that } xb \in \factT \}$. A \emph{supermaximal  right extension} is a right extension that is not a proper suffix of any other right extension. The set of supermaximal right extensions of $T$ is $\SRE(T) =\{x\in E_T\,|\, \forall y \in E_T, y=zx \Rightarrow y=x\}$.

\begin{definition}A \emph{suffixient set} $\sfxtset\subseteq[1,n]$ for $T[1\dd n]$ is a set of positions of $T$ such that every supermaximal right extension $x \in \SRE(T)$ appears as a suffix of some prefix $T[1\dd i]$ with $i \in \sfxtset$.
\end{definition}

\begin{definition}
  The measure $\chi(T)$ is the size of a smallest suffixient set for $T$.
\end{definition}

It is known that $\chi(T)=|\SRE(T)|$ \cite{cenzato2025suffixientarraysnewefficient}.

\begin{definition} \label{def:sarray}
  A \emph{suffixient array} $A$ is a co-lexicographically sorted suffixient set. That is, if $S$ is a suffixient set, then $A[j]=i$ if and only if $T[1\dd i]$ is the $j$-th prefix on the set $\{T[1\dd i]\,|\, i\in S\}$ under $\colt$.
\end{definition}

\begin{example}
Let $T = 010\underline{0}10\underline{1}0\underline{\dol}$, where the underlined positions correspond to a smallest suffixient set $S=\{4,7,9\}$. The prefixes corresponding to those positions satisfy $T[1\dd 9] \colt T[1\dd 4] \colt T[1 \dd 7]$, hence a suffixient array is $A = [9,4,7]$.
\end{example}

\subsection{Computational model}

We assume a computer word of size $w=\Omega(\log n)$, where $n = |T|$. The \emph{packed representation} of a text $T \in [0,\sigma)^n$ is a list obtained by storing $\Theta(\log_\sigma n)$ letters per machine word, thus representing $T$ in $O(n/\log_{\sigma}n)=O(n \log \sigma/\log n)$ machine words. If $T$ is given in the packed representation, we say that $T$ is a \emph{packed string}.
We rely on the following results on packed strings, which yield sublinear times on small enough alphabets, that is, when $\log\sigma = o(\sqrt{\log n})$.

\begin{theorem}[{\cite[Thm.~5.4]{sublinear_BWT_construction_fullversion}}]LCE queries in a text $T\in [0, \sigma)^n$ with $\sigma=n^{O(1)}$ can be answered in $O(1)$ time after $O(n/\log_\sigma n)$-time preprocessing of the packed representation of $T$. 
\end{theorem}

\begin{theorem}[{\cite[Thm.~6.3]{sublinear_BWT_construction_fullversion}}]Given a packed text $T\in [0, \sigma)^n$, the BWT of $T$ can be constructed in $O(n\log \sigma /\sqrt{\log n})$ time and $O(n/\log_\sigma n)$ space.
\end{theorem}

\begin{theorem}[{\cite[Prop.~5.29 and 5.30]{sublinear_suffix_array_construction}}]Let $\epsilon > 0$ be any constant. Given a packed text $T\in [0, \sigma)^n$, we can construct a data structure in $O(n\log \sigma / \sqrt{\log n})$ time and $O(n/\log_\sigma n)$ working space, that can compute $\SA[i]$ and $\ISA[i]$ for any given index $i$ in $O(\log^\epsilon n)$ time.
\end{theorem}

\section{Computing the suffixient array in linear time}

In this section, we present a linear time algorithm that computes a smallest suffixient array starting from $\BWT, \SA, \LCP,$ and a longest common extension data structure  $\LCER$ on the reversed text. Although this algorithm is unlikely to be faster than other state-of-the-art algorithms, it is conceptually very simple, and can be sped up by using proper sublinear-space data structures. In the following sections we show how to exploit this property in order to get an asymptotic improvement over state-of-the-art algorithms on most compressible texts.

We start by introducing a set of positions of $T$ that is suffixient (though not minimal in general), characterized by the following lemma.

\begin{lemma}\label{lem:breaks}
Let $xa$ be a supermaximal right extension of $T$, for $a \in \Sigma$. Then, there exists $i\in[1,n]$ such that $T[\SA[i]..\SA[i]+|x|]=xa$ and at least one of the following conditions is satisfied:
\begin{itemize}
\item $\BWT[i-1]\neq \BWT[i]$ and $\LCE(\SA[i-1],\SA[i])=|x|$ or
\item $\BWT[i+1]\neq \BWT[i]$ and $\LCE(\SA[i+1],\SA[i])=|x|$.
\end{itemize}
\end{lemma}
\begin{proof}
As $xa$ is a supermaximal extension of $T$, there exists $b \in \Sigma$ such that $b \ne a$ and $xa,xb \in \factT$. Let us consider an integer interval $U = \{i\,:\,T[\SA[i]..\SA[i]+|x|]=xa\}$ and let $j$ be any index such that $T[\SA[j]..\SA[j]+|x|]=xb$.

Assume first that $j<k$ for all $k \in U$. Then let $i=\min U$. By the definition of $U$, $T[\SA[i]..\SA[i]+|x|]=xa$ holds. By the lexicographic order, we have $\LCE(\SA[i-1],\SA[i]) \ge \LCE(\SA[j],\SA[i])=|x|$ and $i-1 \not\in U$, so $T[\SA[i-1]+|x|] = c \neq a$ and thus $\LCE(\SA[i-1],\SA[i])=|x|$. Finally, we have $\BWT[i-1]\neq \BWT[i]$ as otherwise $xa$ would be a suffix of the string $\BWT[i] xa \in \factT$, which is a right extension because $\BWT[i-1] x c = \BWT[i] x c \in \factT$.

The case $j>k$ for all $k \in U$ is symmetric; we then have $\LCE(\SA[i+1],\SA[i])=|x|$ and $\BWT[i+1]\neq \BWT[i]$ for $i = \max U$ (so $T[\SA[i]..\SA[i]+|x|]=xa$).
\qed
\end{proof}

\renewcommand{\S}{\mathcal{S}}

\begin{definition} \label{def:S}
Let $\S$ be the following set of substrings of $T$:
$$\S = \{T[\SA[i-1]\dd \SA[i-1]+\LCP[i]],T[\SA[i]\dd \SA[i]+\LCP[i]]\,:\, \BWT[i-1]\neq \BWT[i]\}.$$
\end{definition}
See \cref{fig:main_algorithm}. Note that Lemma~\ref{lem:breaks} implies the following observation.

\begin{observation}
The set 
$\{\SA[i-1]+\LCP[i],\SA[i]+\LCP[i]\,:\, \BWT[i-1]\neq \BWT[i]\}$
is suffixient.    
\end{observation}

\begin{lemma}\label{lem:dom}
A string $u$ is a supermaximal right extension of $T$ if and only if $u \in \S$ and there is no string $v \in \S$ such that $u$ is a proper suffix of $v$.
\end{lemma}
\begin{proof}
By the definition of $\S$, each element of $\S$ is a right extension of $T$.

$(\Rightarrow)$ If $u$ is a proper suffix of a right extension $v \in \S$, then $u$ is not supermaximal.

$(\Leftarrow)$ Assume $u \in \S$ is not a proper suffix of any string $v \in \S$. Then, 
by \cref{lem:breaks}, it is not a proper suffix of any supermaximal right extension. Now assume that $u$ is a proper suffix of a right extension $v_1$ that is not supermaximal. Since $v_1$ is not supermaximal, it is a proper suffix of a longer right extension $v_2$, and so on. We continue this process until we reach a supermaximal string $v$ that has all the previous ones---hence, $u$---as suffixes. \qed
\end{proof}

By \cref{lem:dom}, to obtain a suffixient set of the smallest cardinality, it suffices to remove the elements of $\S$ that correspond to proper suffixes of another element of $\S$. Checking this condition is equivalent to checking if an element is a proper prefix of another, when both are read backwards. 

Let $\rev{\S}$ be the set of reversals of substrings in $\S$, which we treat as substrings of $\TR$. We say that a substring $\TR[n+1-j-\ell\dd n+1-j]$ is \emph{reverse supermaximal} if $T[j\dd j+\ell]$ is supermaximal in $T$. By \cref{lem:breaks}, every reverse supermaximal substring is contained in $\rev{\S}$. Our goal is to filter out the remaining substrings from $\rev{\S}$. The next corollary follows from \cref{lem:dom}.

\begin{corollary}\label{cor:dom}
String $X \in \rev{\S}$ is reverse supermaximal if and only if there is no other substring $Y \in \rev{\S}$ that is prefixed by $X$.
\end{corollary}

This leads to the problem given next. In the problem, we remove strings being proper prefixes of other strings and leave one string from many equal copies.

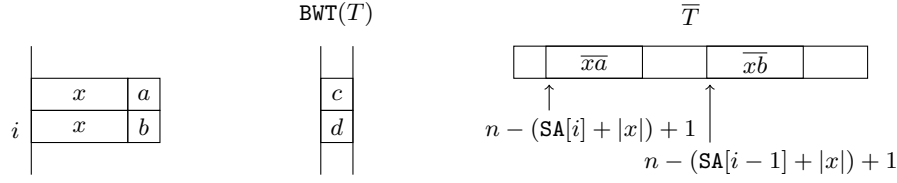
\begin{figure}[t]
\center
\begin{tikzpicture}[scale=0.85]
\draw[-] (0,7) to (0,9);
\draw[-] (4.5,7) to (4.5,9);
\draw[-] (5,7) to (5,9);
\node at (4.75,9.5) {$\BWT(T)$};

\node at (-0.25,7.70) {$i$};
\draw (0,8) rectangle (1.5,8.5) node[pos=.5] {$x$};
\draw (0,7.5) rectangle (1.5,8) node[pos=.5] {$x$};
\draw (1.5,8) rectangle (2,8.5) node[pos=.5] {$a$};
\draw (1.5,7.5) rectangle (2,8) node[pos=.5] {$b$};
\draw (4.5,8) rectangle (5,8.5) node[pos=.5] {$c$};
\draw (4.5,7.5) rectangle (5,8) node[pos=.5] {$d$};

\begin{scope}[xshift=1.5cm]
\draw (6,9) rectangle (11.5,8.5) node[pos=.5] {};
\draw (6.5,9) rectangle (8,8.5) node[pos=.5] {$\rev{xa}$};
\draw (9,9) rectangle (10.5,8.5) node[pos=.5] {$\rev{xb}$};
\node at (7.2,7.7) {$n-(\SA[i]+|x|)+1$};
\node at (10,7.2) {$n-(\SA[i-1]+|x|)+1$};
\node at (8.75,9.5) {$\TR$};
\draw[->] (6.55,8) to (6.55,8.4);
\draw[->] (9.05,7.5) to (9.05,8.4);
\end{scope}
\end{tikzpicture}
\caption{Processing the BWT run-break at position $i$ of the set $\S$. Substrings $T[\SA[i]\dd \SA[i] +|x|]$ and $T[\SA[i-1]\dd \SA[i-1] +|x|]$ are candidates to be supermaximal extensions of $T$. We add the pairs $(n-(\SA[i]+|x|)+1, |x|+1)$ and $(n-(\SA[i-1]+|x|)+1, |x|+1)$ to the multiset $\rev{\S}$, which represents the starting positions and length of the occurrences of those reversed candidates in $\rev{T}$. }\label{fig:main_algorithm}
\end{figure}

\paragraph{Deleting prefixes among a set of substrings.}

We define the following auxiliary problem. Among equal substrings, we retain an arbitrary one.

\newcommand{\argmax}{\mathtt{argmax}}
\defproblem{Deleting Prefixes}{A sequence $T[s_1 \dd e_1),\ldots,T[s_t \dd e_t)$ of substrings of $T$}{List in lexicographic order all different substrings $T[s_i \dd e_i)$ such that $T[s_i \dd e_i)$ is not a proper prefix of $T[s_j \dd e_j)$ for all $j \in [1,t] \setminus \{i\}$}

\begin{example}
If the substrings are $\mathtt{010},\mathtt{01},\mathtt{1010},\mathtt{010},\mathtt{1},\mathtt{101}$, the output should contain substrings $\mathtt{010},\mathtt{1010}$. In the problem, substrings are represented in $O(1)$ space as position pairs.
\end{example}

This problem can be solved using weighted ancestor queries on the suffix tree~\cite{belazzougui_et_al:LIPIcs.CPM.2021.8}. Below we show an arguably simpler algorithm. More importantly, in the next section we show how to make it work in sublinear time (for sufficiently small $t$). Notably, our algorithm does not require to sort the substrings lexicographically; it suffices to sort the corresponding suffixes of $T$.

We say that index $i$ is a (proper) prefix of (or is equivalent to) index $j$ if $T[s_i \dd e_i)$ is a (proper) prefix of $T[s_j \dd e_j)$ (or $T[s_i \dd e_i)=T[s_j \dd e_j)$, resp.).

\newcommand{\ggood}{\mathit{good}}
\newcommand{\nnext}{\mathit{next}}
\newcommand{\pprev}{\mathit{prev}}
In the solution to the problem, we store a decremental doubly-linked list of the indices in $[1,t]$. We consider the substrings in the order of non-decreasing lengths. When an index is considered, we only check if it is a prefix of any of its neighbours in the list and then remove it from the list. Technically, the list is represented as $\nnext/\pprev$ pointers; see \cref{alg:Prefixes}.

\begin{algorithm}[h]
\caption{Solution to \textsc{Deleting Prefixes}}\label{alg:Prefixes}
Reorder the substrings so that $T[s_1 \dd n] \le \cdots \le T[s_t \dd n]$\;
\For{$i\gets 1$ \KwSty{to} $t-1$}{
    $\nnext[i]\gets i+1$\;
    $\pprev[i+1]\gets i$\;}
$\nnext[t] \gets \pprev[1] \gets \infty$\;
\lFor{$i \gets 1$ \KwSty{to} $t$}{$\ggood[i] \gets$ \KwSty{true}}
\ForEach{$(e_i-s_i,i)$ in non-decreasing order}{
    \ForEach{$j \in \{\pprev[i],\nnext[i]\}$}{
        \If{$j \ne \infty$ \KwSty{and} $\lcp(T[s_i \dd e_i),T[s_j \dd e_j)) = e_i-s_i$}{
            $\ggood[i] \gets $ \KwSty{false}\;
        }
    }
    $\nnext[\pprev[i]] \gets \nnext[i]$\;
    $\pprev[\nnext[i]] \gets \pprev[i]$\;
}
\Return{$\{i \in [1,t]\,:\,\ggood[i]=\mathbf{true}\}$}\;
\end{algorithm}

\begin{lemma}\label{lem:radix_sort}
$m$ integers in $[0, n]$ can be sorted in $O(\sqrt{n}+m)$ time.
\end{lemma}
\begin{proof}
We treat an integer $x \in [0, n)$ as a pair $(\lfloor x/\lfloor \sqrt{n} \rfloor \rfloor,x \bmod \lfloor \sqrt{n} \rfloor)$. We sort the resulting $m$ pairs using Radix Sort in $O(\sqrt{n}+m)$ time~\cite{DBLP:books/daglib/0023376}.
\qed\end{proof}

\begin{lemma}\label{lem:Prefixes}
\cref{alg:Prefixes} solves the \textsc{Deleting Prefixes} problem in $O(\min(t+\sqrt{n},t \log t)+tq)$ time provided that any $\LCE$ and $\ISA$ query on $T$ is answered in $O(q)$ time.\end{lemma}
\begin{proof}
\emph{Correctness:} First we show that $\ggood[i]=\mathbf{false}$ at the conclusion of the algorithm if and only if for some $j \in [1,t] \setminus \{i\}$, (1) $i$ is a proper prefix of $j$ or (2) $j$ is equivalent to $i$ and $j>i$.

$(\Rightarrow)$ If $\ggood[i]=\mathbf{false}$, let $j$ be the neighbour of $i$ in the current list that caused this assignment. Naturally, we have $j \in [1,t] \setminus \{i\}$. If $i$ is a proper prefix of $j$, condition (1) is satisfied. If $i$ is equivalent to $j$, we have $e_i-s_i=e_j-s_j$, so $j>i$ by the order in which the list is processed. Hence, condition (2) is satisfied.

$(\Leftarrow)$ Assume that $i,j \in [1,t]$, with $i \ne j$, satisfy (1) $i$ is a proper prefix of $j$ or (2) $j$ is equivalent to $i$ and $j>i$. In either case, when index $i$ is considered in the foreach-loop, index $j$ is still present in the list. In case (1), we assume that $j>i$, as the opposite case is symmetric. Let $i_1=i,i_2,\ldots,i_{w-1},i_w=j$ be all indices in between those two in the list at this moment of the algorithm. Each of the substrings $T[s_{i_l} \dd e_{i_l})$ for $l \in [1,w]$ has length at least $e_i-s_i$, so each of them has $T[s_i \dd e_i)$ as a prefix, as otherwise $T[s_1 \dd n] \le \cdots \le T[s_t \dd n]$ would not be satisfied. Hence, $\ggood[i]=\mathbf{false}$ because of $i_2=\nnext[i]$.

Now we show that substrings are reported in lexicographic order. Assume that consecutive indices $i,i'$ such that $i<i'$ are present in the output of the algorithm. As both indices are in the output, none of the strings $u:=T[s_i \dd e_i]$, $v:=T[s_{i'} \dd e_{i'}]$ is a prefix of the other (in particular, $u \ne v$). Hence, there is a position $q \in [1,\min(|u|,|v|)]$ such that $u[1 \dd q) = v[1 \dd q)$ and $u[q] \ne v[q]$. We have $T[s_i \dd n] \le T[s_{i'} \dd n]$, so $u[q] < v[q]$, i.e., $u<v$.

\emph{Complexity:} Sorting suffixes $T[s_i \dd n]$ by the values $\ISA[s_i]$ as well as values $e_i-s_i$ takes $O(\sqrt{n}+t)$ time by \cref{lem:radix_sort} or $O(t \log t)$ time using merge sort. In addition, the algorithm performs $O(t)$ operations and $O(t)$ $\LCP$ (and $\ISA$) queries. Indeed, an $\lcp$ query on substrings of $T$ can be answered using an $\LCE$ query on the corresponding prefixes of $T$ and a comparison with the substrings' lengths.
\qed\end{proof}

\paragraph{Wrap-up.}

Let us analyze the time complexity of the algorithm. 
We compute the $\BWT$, $\SA$ and $\ISA$ for $T$ and $\ISAR$ for $\TR$ in $O(n)$ time \cite{10.1145/1217856.1217858}. We compute in $O(n)$ time data structures for answering $\LCE$ queries on $T$ and $\LCER$ queries on $\TR$ in $O(1)$ time~\cite{10.1007/10719839_9}. This allows us to compute the set $\S$, and hence the set $\rev{\S}$, in $O(n)$ time. We then use \cref{alg:Prefixes} on $\TR$ to identify reverse supermaximal substrings in $\rev{\S}$. By \cref{lem:Prefixes}, the algorithm works in $O(n)$ time.

By definition, if $\TR[n+1-(j+\ell)\dd n+1-j)$ is such a substring, then $T[j \dd j+\ell)$ is a supermaximal right extension of $T$. By taking all indices $j+\ell-1$ over all such pairs, we obtain a suffixient set of the smallest cardinality.

In the output to the \textsc{Deleting Prefixes} problem, the substrings (of $\TR$) are sorted according to lexicographic order, so they represent a co-lexicographic order of the original substrings of $T$. We leave one copy from each set of equal substrings to ensure that the size of the suffixient array is smallest. Thus we have computed a smallest suffixient array of $T$ in the desired time and space complexity, recall Def.~\ref{def:sarray}.

\paragraph{Building in compressed space.}

As a first application of our algorithm, consider the particularly interesting scenario of building the suffixient array within repetitive-aware space. Prezza \cite{Pre16} showed how to build the run-length-encoded $\BWT$ within $O(r)$ space in time $O(n\log r)$, providing access to the $\BWT$ in time $O(\log\log(n/r))$. Gagie et al.~\cite{GNP20} extend this construction to add structures that compute $\SA$, $\ISA$, and $\LCP$ in time $O(\log(n/r))$. Their construction uses $O(r\log(n/r))$ space and $O(n(\log r + \log\log (n/r)))$ time.

To apply our algorithm, we build those structures for both $T$ and $\TR$. We then compute $\S$ and $\rev{\S}$ in time $O(r\log(n/r))$, whereas the prefix deletion phase takes $O(r \log r + r\log(n/\rR))$ time. The total time is then $O(r\log n) \subseteq O(n\log r)$ and the extra working space is $O(r)$. (We improve this result later, in Corollary~\ref{cor:compressed}.)

\begin{theorem} \label{thm:compressed}
A suffixient array of a string $T \in [0, \sigma)^n$, where the $\BWT$ of $T$ has $r$ runs and that of $\TR$ has $\rR$ runs, can be computed in $O(n(\log r^*+\log\log(n/r^*)))$ time using $O(r^*\log (n/r^*))$ working space, where $r^* = \max(r,\rR)$.
\end{theorem}

\section{An $O(p(n,\sigma) {+} r\log^\epsilon n)$ time algorithm for computing $\chi$}

Let $p(n,\sigma) = \frac{n\log \sigma}{\sqrt{\log n}}$ and $p'(n,\sigma) = \frac{n \log \sigma}{\log n}$.
In this section, we present an algorithm that, given a packed string $T$, computes a smallest suffixient array in $O(p(n,\sigma) + r\log^\epsilon n)$ time for any desired constant $\epsilon > 0$, and uses $O(p'(n,\sigma)+r)$ working space. We use the algorithm from the previous section but replace the specific data structures with their sublinear-time counterparts.

\begin{theorem}\label{thm:main}
A smallest suffixient array of a packed string $T \in [0, \sigma)^n$ can be computed in $O(p(n,\sigma) + r\log^\epsilon n)$ time for any desired constant $\epsilon > 0$, using $O(p'(n,\sigma)+r)$ working space. As a by-product, the algorithm computes $\SRE(T)$.
\end{theorem}
\begin{proof}
We compute the (packed) $\BWT$ of $T$ in $O(p(n,\sigma))$ time and $O(p'(n,\sigma))$ space~\cite{sublinear_BWT_construction_fullversion}.
Since the $\BWT$ is packed in $O(p'(n,\sigma))$ space, we can identify all run breaks, in each machine word separately, in $O(p'(n,\sigma)+r)$ time: We precompute in time $O(\sqrt{n}\log n)$ a table that, for every sequence of packed symbols fitting in $\lfloor \frac{1}{2}\log n\rfloor$ bits, records the run breaks inside the sequence.

We compute the packed representation of $\TR$. To this end, we reverse the list of words representing $T$ and also need to reverse the symbols within each word. To this end, we precompute in time $O(\sqrt{n} \log n)$ a table that stores the reverse of every sequence of packed symbols fitting in $\lfloor \frac12 \log n \rfloor$ bits.
Now after $O(p(n,\sigma))$ time and $O(p'(n,\sigma))$ space preprocessing, we can compute $\SA[i]$ and $\ISAR[i]$ for any $i$ in $O(\log^\epsilon n)$ time~\cite{sublinear_suffix_array_construction}, and after $O(p'(n,\sigma))$ time preprocessing, the $\LCE$ of any two positions in $T$ and in $\TR$ can be computed in $O(1)$ time~\cite{sublinear_BWT_construction_fullversion}.

With these data structures, the set $\S$, and hence the set $\rev{\S}$, can be computed in $O(p(n,\sigma)+r\log^\epsilon n)$ time and $O(p'(n,\sigma)+r)$ space (via $2r$ $\SA$, $\ISAR$, and $\LCE$ queries on the $\BWT$ run limits). This allows creating the lists so that \cref{alg:Prefixes} runs  in $O(r+\sqrt{n})$ time, by \cref{lem:Prefixes}. See \cref{alg:sublinear_from_text}.
\qed
\end{proof}

\begin{algorithm}[t]
\caption{Computing a smallest suffixient array in sublinear time}\label{alg:sublinear_from_text}
\SetKwComment{Comment}{/* }{ */}
\KwData{A packed representation $L_T$ of a text $T$}
\KwResult{A smallest suffixient array $A$ for $T$.}
$L_{\TR} \gets \mathtt{reverse}(L_{T})$\;
$\LCER \gets \LCER(L_{\TR})$\;
\For{each run-break position $i \in [\rho_1, \rho_{r-1}]$}{
    $\rev{\S} \gets \rev{\S} \cup \{(n+1-(\SA[i-1]+\LCP[i]), \LCP[i]+1)\}$\;
    $\rev{\S} \gets \rev{\S} \cup \{(n+1-(\SA[i]+\LCP[i]), \LCP[i]+1)\}$\;
}
$A\gets \mathtt{DeletingPrefixes}(L_{\TR},\LCER,\rev{\S})$\;
\Return{$A$}
\end{algorithm}

\paragraph{An always-sublinear-time approximation.}

When $\log\sigma=o(\sqrt{\log n})$, Theorem~\ref{thm:main} yields $o(n)$ construction time if there is constant $\epsilon>0$ such that $r = o(n/\log^\epsilon n)$. We can ensure $o(n)$ time for any $r$ by building an approximation, for any fixed $\epsilon$: We first build $\BWT$ in $o(n)$ time, which yields $r$. If $r \le n/\log^{\epsilon/2} n$, we can just apply Theorem~\ref{thm:main}. Otherwise, we return instead the suffixient set defined by the endpoints of $\S$ (Def.~\ref{def:S}), which is of size $2r$ and available from $\BWT$, $\SA$, and $\LCP$, all of which are built in $o(n)$ time. This size is an $O(\log^{1+\epsilon} n)$-approximation to the optimal $\chi$, because $2r = O(n)$ and $\chi \ge \delta = \Omega(r /(\log \delta \log(n/\delta)))$ \cite{kempa2020resolution}, where $\delta$ is the normalized substring complexity \cite{RRRS13}. When $r = \Omega(n/\log^{\epsilon/2} n)$, since $\delta = \Omega(r/\log^2 n)$ \cite{kempa2020resolution}, it holds $n/\delta = O((n/r)\log^2 n) = O(\textrm{polylog}\,n)$. Thus $\chi = \Omega(r/(\log\delta\log(\textrm{polylog}\,n))) = \Omega(n/(\log^{\epsilon/2} n \log n \log\log n)) = \Omega(n/\log^{1+\epsilon} n)$.

Further, if we switch to just returning the endpoints of Def.~\ref{def:S} when we detect that $r > (n\log\sigma/\sqrt{\log n})/\log^{\epsilon/2} n$, then we ensure total time $O(p(n,\sigma))$, at the price of worsening the approximation to $O(\log^{3/2+\epsilon} n)$.

\begin{corollary}\label{cor:approx}
A suffixient array of a string $T \in [0, \sigma)^n$, of size $O(\chi \log^{1+\epsilon} n)$ for any desired constant $\epsilon>0$, can be computed in $o(n)$ time using $O(p'(n,\sigma))$ working space plus the output size. The suffixient array is of size $O(\chi)$ whenever $r = O(n/\log^\epsilon n)$. Further, we can compute in $O(p(n,\sigma))$ time a suffixient array of size $O(\chi \log^{3/2+\epsilon} n)$, which is of size $\chi$ when $r = O(n\log\sigma/\log^{1/2+\epsilon} n)$.
\end{corollary}

\paragraph{Replacing the $r$ factor by $\min(r,\rR)$.}

Cenzato et al. \cite[Algorithm 6]{cenzato2025suffixientarraysnewefficient} shows how to construct a smallest suffixient array for $T$ in linear time, starting from the arrays $\BWTR$, $\SAR$, and $\LCPR$, of the reversed text. We adapt their algorithm using the data structures of Kempa and Kociumaka \cite{sublinear_BWT_construction_fullversion}, obtaining the following result (see Appendix~\ref{sec:adaptation_cenzato}).

\begin{theorem}\label{thm:main_reverse}
A smallest suffixient array of a string $T \in [0, \sigma)^n$ can be computed in $O(p(n,\sigma) + \rR\log^\epsilon n)$ time for any constant $\epsilon > 0$, using $O(p'(n,\sigma)+\rR)$ working space. 
\end{theorem}

By combining \cref{thm:main} and \cref{thm:main_reverse}, and because $r$ and $\rR$ can be computed in $O(p(n,\sigma))$ time, we obtain the following.

\begin{corollary}
A smallest suffixient array of a string $T \in [0, \sigma)^n$ can be computed in $O(p(n,\sigma) + \min(r,\rR)\log^\epsilon n)$ time for any desired constant $\epsilon > 0$, using $O(p'(n,\sigma)+\min(r,\rR))$ working space. 
\end{corollary}

Finally, repeating the same steps of Theorem~\ref{thm:compressed} over this algorithm, which uses only the structures on $\TR$, we have the following improved result.

\begin{corollary} \label{cor:compressed}
A suffixient array of a string $T \in [0, \sigma)^n$, where the $\BWT$ of $\TR$ has $\rR$ runs, can be computed in $O(n(\log \rR+\log\log(n/\rR)))$ time using $O(\rR\log (n/\rR))$ working space.
\end{corollary}

\section{Computing $O(\chi(T))$-Sized String Representation}

Shibata and Bannai~\cite{chi_reachable} showed recently that a string $T$ can be represented in $O(\chi(T))$ space. They presented (in Thm.~15) an $O(n)$-time algorithm computing this representation (if $T$ is over a linear-time-sortable alphabet). We show that an $O(\chi(T))$-sized representation of $T$ can be computed in sublinear time (for small enough $\sigma$ and $r$).

The next lemma follows from efficient construction of wavelet trees~\cite{DBLP:conf/soda/BabenkoGKS15}.

\begin{lemma}\label{lem:Ch}
For a packed text $T \in \Sigma^n$, we can compute for each letter $c \in \Sigma$ a position $i$ such that $T[i]=c$ or state that no such position exists in $O(p(n,\sigma)+\sigma \log \sigma)$ time and $O(p'(n,\sigma)+\sigma)$ space. 
\end{lemma}
\begin{proof}
A $\mathit{select}_a(i)$ query for $T$ consists in computing the $i$th leftmost occurrence of character $a \in \Sigma$ in $T$ (or stating that such an occurrence does not exist). A wavelet tree for $T$ allows to answer $\mathit{select}_a$ queries in $O(\log \sigma)$ time~\cite{DBLP:conf/soda/GrossiGV03,DBLP:journals/jda/Navarro14}. Thus queries $\mathit{select}_a(1)$ for all $a \in \Sigma$ can be answered in $O(\sigma \log \sigma)$ time with the data structure. Finally, a wavelet tree can be constructed in $O(p(n,\sigma))$ time~\cite{DBLP:conf/soda/BabenkoGKS15} and occupies $O(p'(n,\sigma))$ space.
\qed\end{proof}

The representation of Shibata and Bannai is a substring equation system. 

\begin{definition}[{\cite{chi_reachable}}]
An instance of a \emph{substring equation system} (SES in short) for a string $T\in\Sigma^n$ is a triple $(n,\mathsf{Eq},\mathsf{Ch})$.
Each element of $\mathsf{Eq}$ is a triple $(i,j,\ell)$ with $i,j \in [1,n]$ and $1\le \ell \le n-\max(i,j)+1$, representing the equation $T[i\dd i+\ell)=T[j\dd j+\ell)$.
Each element of $\mathsf{Ch}$ is a pair $(k,c)$ with $k \in [1,n]$ and $c\in\Sigma$, representing the equation $T[k]=c$.
We say that $(n,\mathsf{Eq},\mathsf{Ch})$ \emph{represents} $T$ if $T$ is the unique string in $\Sigma^n$ satisfying all constraints in $\mathsf{Eq}$ and $\mathsf{Ch}$.
The \emph{size} of the system is defined as $|\mathsf{Eq}|+|\mathsf{Ch}|$.
\end{definition}

We now recall how the SES for $T$ of size $O(\chi(T))$ is constructed \cite{chi_reachable}. The set $\mathsf{Ch}$ is obtained by taking one arbitrary occurrence in $T$ of each distinct character. The set $\mathsf{Eq}$ includes the following equations: for every two $yxa,y'xa' \in \SRE(T)$, where $x$ is the longest common suffix of $yx$ and $y'x$, pick arbitrary occurrences $T[s \dd e)=yxa$ and $T[s' \dd e')=y'xa'$ and create an equation $T[s+|y| \dd s+|yx|)=T[s'+|y'| \dd s'+|y'x'|)$. Shibata and Bannai show how all these equations can be expressed by just $\chi(T)-1$ equations using a construction based on a compacted trie of substrings of $T$. 

\begin{theorem}
For a text $T \in [0, \sigma)^n$ such that $\log \sigma \le \sqrt{\log n}$, an SES of size $O(\chi(T))$ representing $T$ can be computed in $O(p(n,\sigma)+r\log^\epsilon n)$ time and $O(p'(n,\sigma)+r)$ space.
\end{theorem}
\begin{proof}
Since $\sigma \le 2^{\sqrt{\log n}}$, the set $\mathsf{Ch}$ can be computed via \cref{lem:Ch} in $O(p(n,\sigma))$ time and $O(p'(n,\sigma))$ space. For $\mathsf{Eq}$, we present a construction that avoids building the compacted trie (and explicitly sorting the trimmed reversals of supermaximal right extensions).

We compute $\SRE(T)$ using \cref{thm:main}. For each $T[i \dd j] \in \SRE(T)$, we construct a substring $\TR[n+2-j \dd n+1-i] = \rev{(T[i \dd j))}$ (note that the last symbol of the supermaximal right extension is removed). Let $\TR[s_1 \dd e_1),\ldots,\TR[s_t \dd e_t)$ be the resulting substrings, where $t=\chi(T)$. Then the set $\mathsf{Eq}$ is constructed using \cref{alg:Eq} (which resembles \cref{alg:Prefixes}).

\begin{algorithm}[t]
\caption{Construction of $\mathsf{Eq}$}\label{alg:Eq}
Reorder substrings $\TR[s_i \dd e_i)$ so that $\TR[s_1 \dd n] \le \cdots \le \TR[s_t \dd n]$\;
$\mathsf{Eq} \gets \emptyset$\;
\For{$i:=1$ \KwSty{to} $t-1$}{
    $\ell \gets \LCER(s_i,s_{i+1})$\;
    \If{$\ell > 0$}{
        $\mathsf{Eq} \gets \mathsf{Eq} \cup \{(n+2-s_i-\ell,n+2-s_{i+1}-\ell,\ell)\}$\;
    }
}
\Return{$\mathsf{Eq}$}\;
\end{algorithm}

\cref{alg:Eq} works in $O(t+\sqrt{n})$ time---we use \cref{lem:radix_sort} for sorting---after a preprocessing that computes $\ISAR$~\cite{sublinear_suffix_array_construction} and the data structure for answering $\LCER$ queries~\cite{sublinear_BWT_construction_fullversion}. The preprocessing itself takes $O(p(n,\sigma))$ time and $O(p'(n,\sigma))$ space. The returned set satisfies $|\mathsf{Eq}|=\chi(T)-1$. We just need to argue that it expresses all the desired equations. The triple inserted to $\mathsf{Eq}$ corresponds to equation $$T[n+1-(s_i+\ell-1) \dd n+1-s_i]=T[n+1-(s_{i+1}+\ell-1) \dd n+1-s_{i+1}],$$ i.e., $\TR[s_i \dd s_i+\ell)=\TR[s_{i+1} \dd s_{i+1}+\ell)$ where $\ell=\lcp(\TR[s_i \dd n],\TR[s_{i+1} \dd n])$. Thus every triple  inserted into $\mathsf{Eq}$ implies an equation that holds in $\TR$ (hence, in $T$). Let us consider indices $j,k$ such that $1 \le j < k \le t$ and denote $d=\lcp(\TR[s_j \dd e_j),\TR[s_k \dd e_k))$. We need to argue that the equation  $$(*)\,\, \TR[s_j \dd s_j+d)=\TR[s_k \dd s_k+d)$$ is implied by the equations in $\mathsf{Eq}$. Because the suffixes $\TR[s_1\dd n],\ldots,\TR[s_t \dd n]$ are ordered lexicographically, $\TR[s_i \dd s_i+d)=\TR[s_j \dd s_j+d)$ holds for all $i \in [j,k]$. This implies that $\LCER(i,i+1) \ge d$ for all $i \in [j,k-1]$, and equation $(*)$ is implied by the equations created for $i \in [j,k-1]$.
\qed\end{proof}

\section{Future Work}

We have shown how to compute a suffixient array for $T[1\dd n]$ in $O(n\log \sigma/\sqrt{\log n}$ $+\min(r,\rR)\log^\epsilon n)$ time for any $\epsilon > 0$. 

While further reducing the first term hits conditional lower bounds as long as we build on top of the BWT \cite{sublinear_BWT_construction_fullversion}, the second term comes from computing $2\min(r,\rR)$ direct and inverse suffix array queries using structures that can be built in sublinear time \cite{sublinear_BWT_construction_fullversion}. It could be possible to run those queries faster, for example in batch rather than one by one.

We also plan to work further on the line of Theorem~\ref{thm:compressed}, by considering conversions from/to other compressed formats (see, e.g., \cite{GGJNlatin24.3}). It might be possible, for example, to find connections with other compressed indexes that build on right-maximal strings, like CDAWGs \cite{BBHMCE87}.

%
%
\bibliographystyle{splncs04}
\bibliography{bibliography}

\appendix 

\section{Sublinear time adaptation of Cenzato et al.'s algorithm}\label{sec:adaptation_cenzato}

Cenzato et al. \cite[Algorithm 6]{cenzato2025suffixientarraysnewefficient} show how to construct a smallest suffixient array for $T$ in linear time, starting from the  array $\BWTR$, the suffix array $\SAR$, and the LCP array $\LCPR$ of the reversed text. 
Additionally, it uses an array $R$ to keep information about the last position that is candidate to be a super-maximal extension for each $c \in \Sigma$. More precisely, in $R[c].len$ is stored the $\LCPR$ value for the last candidate $c$-run break, in $R[c].pos$ is stored its position in $T$, and $R[c].active$ is a flag telling if $R[c].pos$ was already added to the suffixient set (which case it is set to false) or not.

The algorithm sequentially scans $\BWTR$. For a run $\BWTR[i - 1 - j \dd i - 1] = c^{j + 1}$, it iterates from $k = i - j$ to $i - 1$ and keeps the minimum value within $\LCPR[i - 1 - j \dd k]$ in a variable $m$ and it updates $\LFR[c] = \LFR[c] + 1$. For a run-break $\BWTR[i - 1, i] = ac$, it keeps the minimum value within $\LCPR[i]$ and $m$ and stores it in $m$, updates $\LFR[c] = \LFR[c] + 1$ and calls to $eval$ to evaluate if $m < R[a].len$ (such a case $R[a].pos$ is added to the suffixient set computed so far, $m$ is stored in $R[a].len$ and $R[a].active$ is set to $false$), and it calls to $eval$ to determine if $\LCPR[\LFR[\BWTR[c]]] - 1 < R[c].len$ (such a case $R[c].pos$ is added to the suffixient set, $\LCPR[\LFR[\BWTR[c]]] - 1$ is stored in $R[c].len$, and $R[c].active$ is set to false). Additionally, if $R[\BWTR[i']].len < \LCPR[i]$, then $n - \SAR[i'] + 1$ is stored in $R[\BWTR[i']].pos$, $\LCPR[i]$ is stored in $R[\BWTR[i']].len$ and $R[\BWTR[i']].active$ is set to true (with $i' \in \{i - 1, i\}$). Finally, after traversing the whole $\BWTR$, any active candidate stored in $R$ is added to the suffixient set. Thus, we can group the iterations into traversing within a run to keep the minimum $\LCPR$ value and update $\LFR$, and performing $O(1)$ operations for each run-break.

We modify the algorithm to avoid scanning the $\BWTR$ position by position and just iterate from one run-break to the next one. To achieve this, we can note that, for two consecutive run-breaks $b_{i - 1}$ and $b_i$, the minimum within $\LCPR[b_{i - 1} \dd b_{i}]$ is exactly the length of the prefix shared by from the $(b_{i - 1})$-th suffix to the $b_i$-th suffix of $T$, so instead of iterating from $b_{i - 1}$ to $b_i$ we compute $\LCER(\SAR[b_{i - 1}], \SAR[b_i])$. In addition, if $1 < b_i - b_{i - 1}$ (so, $\BWTR[b_i \dd b_{i - 1} - 1] = c^{b_i - b_{i - 1}}$), we have to count $b_i - b_{i - 1} - 1$ additional occurrences of $c$ to get $\LFR[c]$ to point to $b_i - 1$ as desired (we counted one when we processed $b_{i - 1}$), so we update $\LFR[\BWTR[a]] = \LFR[\BWTR[a]] + (b_i - b_{i - 1} - 1)$. The rest of the algorithm just needs to know the positions of the run-breaks.
\cref{alg:sublinear_from_reversed_text} shows how to run the algorithm with these changes using the data structures of Kempa and Kociumaka \cite{sublinear_BWT_construction_fullversion}.

\begin{algorithm}[t]
\caption{Computing eval}\label{alg:eval}
\KwData{A set of characters $C$, an LCP value $\ell$, the candidate list $R$, and a suffixient set $S$}
\KwResult{The updated suffixient set $S$}
\ForEach{$c\in C$}{
\If{$\ell< R[c].len$}{
    \If{$R[c].active$}{
        $S\gets S\cup\{R[c].pos\}$\;
    }
    $R[c] \gets \{l,0,\false\}$\;
}
}
\end{algorithm}

\begin{algorithm}[t]
\caption{Adaptation of Cenzato et al.'s linear-time algorithm \cite[Algorithm 6]{cenzato2025suffixientarraysnewefficient} for computing a smallest suffixient array}\label{alg:sublinear_from_reversed_text}
\SetKwComment{Comment}{/* }{ */}
\KwData{A text $T$ and data structures for $\BWTR,\SAR$ and $\LCPR$}
\KwResult{A smallest suffixient set $S$ for $T$.}
{$R[1,\sigma] \gets (len \gets -1, pos \gets 0, active \gets false) \times \sigma$}\;
{$S \gets \emptyset$}\;
{$\LFR[a_1]\gets 0$}\;
\lFor{$i=2,\ldots, \sigma$}{$\LFR[a_i]\gets \LFR(a_{i-1})+occ(T,a_{i-1})$}
{$\LFR[\BWTR[1]] \gets \LFR[\BWTR[1]] +1 $}\;
\For{$i=1,\dots, \rR$}
{
    $\LFR[\BWTR[b_i-1]]\gets \LFR[\BWTR[b_i-1]] + (b_i-b_{i-1}-1)$\Comment*[r]{Assuming $b_0 = 1$}
    $\LFR[\BWTR[b_i]]\gets \LFR[\BWTR[b_i]] + 1$\;
    $m \gets \LCER(\SAR[b_{i-1}],\SAR[b_i])$\Comment*[r]{$m \gets \min_{j\in[b_{i-1},b_i]}\{\LCPR[j]\}$}
    $eval(\BWTR[b_i - 1], m, R, S)$\;
    \lIf{$R[\BWTR[b_i]].len \neq -1$}
    {
        $eval(\BWTR[b_i], \LCPR[\LFR[\BWTR[b_i]]] - 1, R, S)$}
    \For{$b' \in \{b_i - 1, b_i\}$}
    {
        \If{$R[\BWTR[b']].len < \LCPR[b_i]$}{$R[\BWTR[b'] \gets (\LCPR[b_i], n - \SAR[b'] + 1, true)$\;}
    }
}
$eval(\Sigma,-1,R,S)$
\end{algorithm}

The arrays $R$ and $\LFR$ use $O(\sigma)$ space. The sublinear data structures $\BWTR,\LCPR,$ and $\SAR$ take $O(p'(n,\sigma))$ space in total. Thus, the space consumption of the algorithm adds up to $O(p'(n,\sigma) + \sigma) = O(p'(n,\sigma))$.

Regarding running-time, to compute $occ(T,a_{i})$, we start by setting $occ(T, a_i) = 0$ for each $a_i$ and $occ(T, \BWTR[1]) = 1$. For simplicity, we assume $b_0 = 1$. Now, for each run-break $b_i$, we set $occ(T, \BWTR[b_i - 1]) = occ(T, \BWTR[b_i - 1]) + ((b_i - 1) - b_{i - 1})$ (so we are counting $(b_i - 1) - b_{i - 1}$ additional occurrences if $\BWTR[b_{i - 1}, b_i - 1] = c^{(b_i - 1) - b_{i - 1} + 1}$, and none otherwise) and $occ(T, \BWTR[b_i]) = occ(T, \BWTR[b_i]) + 1$. If $b_{\rR} < n$ (so, $\BWTR[b_{\rR}, n] = c^{n - b_{\rR} + 1}$), then we set $occ(T, \BWTR[b_{\rR}]) = occ(T, \BWTR[b_{\rR}]) + (n - b_{\rR})$ to count the occurrences within that last run. Since we perform $O(1)$ operations on each run-break, this step takes $O(\rR)$ time. The call to $eval$ on Line 10 takes $O(1)$ time, because the set given as the first argument is of size 1. The last call to $eval$ on Line 15 takes $O(\sigma)$ time as the set given has size $O(\sigma)$. 
The total running-time and space consumption are then dominated by that of building the base structures \cite{sublinear_BWT_construction_fullversion} on $\TR$ and using them $O(\rR)$ times.

This algorithm computes only a smallest suffixient set. We can sort this suffixient set in $O(\sqrt{n} + \chi)$ time (using \cref{lem:radix_sort}) by enhancing the prefixes with their co-lexicographical order, which can be retrieved when inserting elements in $S$ using $\rev{\ISA}$ queries. This does not increase (asymptotically) the space or time of the algorithm.
\end{document}